\documentclass{siamonline190516}
\usepackage{graphicx,graphics, pdfpages,subfig}
\usepackage[sort]{natbib}
\usepackage{amsmath, amssymb, enumerate, color}

\newcommand{\e}{\mathbb{E}}

\newcommand{\EE}{\mathbb{E}}
\newcommand{\PP}{\mathbb{P}}

\newcommand{\Var}{\mathbb{V}}

\newcommand{\hU}{\widehat{U}}
\newcommand{\hS}{\widehat{S}}

\newcommand{\tr}{\mathrm{tr}}

\renewcommand{\Pr}{\mathbb{P}}

\begin{document}\sloppy

\newsiamthm{claim}{Claim}
\newsiamthm{result}{Theorem}
\newsiamremark{remark}{Remark}
\newsiamremark{assumption}{Assumption}
\newsiamremark{model}{Model}

\title{{\bf Numerical Tolerance for Spectral Decompositions of Random Matrices}}

\author{Avanti~Athreya, Michael~Kane, Bryan~Lewis, Zachary~Lubberts, Vince~Lyzinski, Youngser~Park, Carey~E.~Priebe, \and Minh~Tang}

\date{}
\maketitle
\begin{abstract}
 We precisely quantify the impact of statistical error in the quality of a numerical approximation to a random matrix eigendecomposition, and under mild conditions, we use this to introduce an optimal numerical tolerance for residual error in spectral decompositions of random matrices. We demonstrate that terminating an eigendecomposition algorithm when the numerical error and statistical error are of the same order results in computational savings with no loss of accuracy. We also repair a flaw in a ubiquitous termination condition, one in wide employ in several computational linear algebra implementations.  We illustrate the practical consequences of our stopping criterion with an analysis of simulated and real networks. Our theoretical results and real-data examples establish that the tradeoff between statistical and numerical error is of significant import for data science.
\end{abstract}
\section{Introduction}
\label{S:Intro}
The rapid and accurate computation of the eigenvalues and eigenvectors of a matrix is of universal importance in mathematics, statistics, and engineering. In practice, of course, numerical methods to compute such spectral decompositions necessarily involve the imposition of a stopping criterion at which a given linear algebraic algorithm 
 terminates, typically when the residual numerical error is less than some user-specified tolerance. When matrices have random entries, however, there is a second important source of error: the inherent {\em statistical error} (between, say, matrix entries and their mean) in addition to the {\em numerical error} from algorithmic approximations. Both of these sources of error contribute to the discrepancy between computed and theoretical eigendecompositions, and their interplay can be important for determining the optimal termination of an algorithm. Stopping an algorithm when the numerical error is too large can yield unacceptably inaccurate output, but stopping when the numerical error is very small---and is effectively dwarfed by the statistical error---can be computationally expensive without resulting in a meaningful improvement in accuracy.  

We focus on determining an optimal error tolerance for the numerical computation of an eigendecomposition of a random symmetric matrix $A$. Such matrices are ubiquitous in data science, from errorful observations of data to adjacency matrices of random networks. In particular, spectral decompositions of the adjacency and Laplacian matrices of random networks have broad applications, from the distributional convergence of random walks on graphs (\cite{chung_grigoryan_yau_upper}) to the solution of a relaxation of the min-cut problem (\citet{fiedler1973algebraic}).  Further, for random dot product graphs (RDPGs), a particular model which we define and describe in more detail in Section~\ref{S:Simulations}, the spectral decomposition of $A$ can serve as a statistical estimate for underlying graph parameters (see \cite{STFP-2011}).
As described above, some nonzero sampling error is inherent in such procedures; that is, there is random discrepancy between a statistical estimate for a parameter and the true value of the parameter itself. If the order of magnitude of this sampling error is known---for instance, if one can obtain a lower bound for the sampling error incurred when using a spectral decomposition of $A$ to estimate the spectral decomposition of its mean---then the accuracy in the numerical algorithm for the spectral decomposition of $A$ should be weighed against this inherent variability.  There may be little gain in a very careful determination of the eigenvectors of $A$ if these eigenvectors are, with high probability, close to some fixed, nonzero distance from the true model parameters.

  To be precise, suppose $A$ is a symmetric square matrix. Let $S_A$ denote the diagonal matrix of the $d$ largest-magnitude eigenvalues of $A$, and $U_A$ the matrix of corresponding eigenvectors.  Numerical methods for both the singular value decomposition and eigendecomposition---and, in turn, for $S_A$ and $U_A$---abound. Iterative methods, such as the power method, Lanczos factorization, and Rayleigh quotient iteration, to name but a few (see \citet{higham,stewart} for a comprehensive overview) begin with an initialization $\hS_0$ for the eigenvalues or singular values of $A$, and $\hU_0$ for the eigenvector(s), and compute successive updates $\hS_k$ and $\hU_k$ for the eigenvalues and eigenvectors, respectively, until a specified error, often a function of the difference between $A\hU_k$ and $\hU_k\hS_k$, is sufficiently small. For example, many algorithms are set to terminate when the relative error is suitably small: that is, when
\begin{equation}
\label{eq:tol}
\frac{\|A\hU_k- \hU_k\hS_k\|_2}{\|A\|_2}\leq \epsilon.
\end{equation}
 where $\| \cdot \|_2$ denotes the spectral norm.
 Under suitable rank and eigengap assumptions on $A$, convergence of the iterates is guaranteed for generic initializations. When the spectral norm of $A$ is known, the relative error on the left hand side of Eq.\eqref{eq:tol} is directly computable; alternatively, when $\|A\|_2$ must be approximated numerically, it can be replaced by the absolute value of an approximation of the largest-magnitude eigenvalue of $A$. For further details on error analysis and stopping criteria for numerical methods, see \citet{kahan,stewart,arioli_92,arioli_89,rigal_67,hestenes_52}.
 
 

Unfortunately, the stopping criterion in Equation~\eqref{eq:tol} has a known flaw: when the columns of the initialization $\hat{U}_0$ belong to an invariant subspace of $A$ that is not the span of the top $d$ eigenvectors---or, because of numerical error, they appear to belong to such a subspace---then even when this stopping criterion is satisfied, the wrong eigenspace may be approximated. Typically, it is assumed that $U_A^{T} \hat{U}_0$ is nonsingular, but such nonsingularity is frequently unverifiable, and it is unclear how numerical imprecision affects it. We offer the alternative stopping criterion \begin{equation}
\label{eq:tol_alt}
\frac{\|A\hat{U}_k-\hat{U}_k\hat{S}_k\|_2}{\|A\|_2}+\|\hat{S}_k^{-1}\|_2\leq\epsilon,\end{equation} and show that under our model assumptions and choice of $\epsilon$, when this criterion is satisfied, $\|\hat{S}_k^{-1}\|_2$ is of strictly lower order than the relative error, so the stopping criteria in Eq.~\eqref{eq:tol} and Eq.~\eqref{eq:tol_alt} are very similar. However, when the wrong eigenspace is being approximated, our added term will be too large, so it will prevent the algorithm from terminating. Consequently, when this new stopping criterion holds, the algorithmic output approximates the correct eigenspace.

Our main result is the following. Let $A$ be a symmetric, random matrix with independent, uniformly bounded entries. We suppose that the mean matrix $P$ defined by $P_{ij}=\EE[A_{ij}]$ is rank-$d$ and positive semidefinite. Let $\mathcal{S}=[\sigma_{ij}^{2}]_{ij=1}^{n}$ denote the matrix of variances for $A$, so that \begin{equation}\label{eq:vardef}\sigma_{ij}^{2}=\Var(A_{ij})=\EE (A_{ij}-P_{ij})^{2}\text{ for all }1\leq i\leq j\leq n.\end{equation} Let $\epsilon$ be the error tolerance in the numerical algorithm used to compute the rank-$d$ eigendecomposition of $A$. Then under mild assumptions on the eigenvalues of $P$ and the order of the rows of $\mathcal{S}$, we find that $\epsilon$ need not be much smaller than $\|A\|_2^{-1/2}$ before further reduction will not improve the accuracy of the numerical eigendecomposition: see Section~\ref{S:Opt_Tol}. We also show in Section~\ref{S:Simulations} via numerical simulations on real and simulated data that this choice of $\epsilon$ can lead to considerable computational savings compared to relying on a default algorithmic tolerance, in a problem of estimating the latent positions for a random graph model, and that it does not negatively affect subsequent inference.

The organization of the paper is as follows: In the following Section~\ref{S:setting}, we set our notation, terminology, and random matrix model. In Section~\ref{S:StatisticalErrorBounds}, we state our results on the order of the statistical error in the approximation of the eigendecomposition of $P$ when using the (exact) eigendecomposition of the observed matrix $A$. In Section~\ref{S:NumericalErrorBounds}, we state our results on the order of the numerical error in the approximation of the eigenvectors of $A$ when stopping an algorithm according to condition~\eqref{eq:tol_alt} for a given error tolerance $\epsilon$. In Section~\ref{S:Opt_Tol}, we synthesize these results to obtain an optimal choice of $\epsilon$ in the stopping criterion. Our simulation results are contained in Section~\ref{S:Simulations}, and all proofs are located in the Appendix~\ref{AS:Proofs}.

\section{Setting and notation}
\label{S:setting}

Let $\Omega$ represent our sample space and $\mathbb{P}$ our probability measure. The expectation of a random variable will be denoted by $\mathbb{E}$. If $v$ is a vector, $|v|$ denotes its Euclidean length. For any $n \times n$ real-valued matrix $M$, $M^{\top}$ denotes the transpose, $\tr(M)$ the trace,  $\|M\|_F$ the matrix Frobenius norm, and $\|M\|_2$ the spectral norm. For any symmetric matrix $M$ with some number of non-zero eigenvalues, let $\lambda_1(M), \cdots, \lambda_d(M)$ denote, in decreasing order, the $d$ eigenvalues of $M$ with largest magnitude (so $|\lambda_1(M)|\geq |\lambda_2(M)|$, etc).

We now describe the random matrix model we will consider throughout this work.

\begin{model} (LPSM)\label{ass:model}
We say a random symmetric matrix $A$ of order $n$ with independent entries $A_{i,j}\in[\alpha_{ij},\alpha_{ij}+\beta]$ follows a {\em low-rank, positive semidefinite mean (LPSM) model} (or that $A$ is an LPSM) if the mean matrix $P=\EE[A]$ is positive semidefinite and of rank $d$, where $d$ is fixed.
\end{model}

Recalling that $\sigma^2_{ij}=\Var(A_{ij})$ as in Eq. \eqref{eq:vardef}, we define the following quantities: \begin{equation}\label{eq:sigmandef}\sigma^{2}(n):=\max_{i}\sum_{j=1}^{n}\sigma_{ij}^{2},\text{ and } \mu^{2}(n):=\min_{i}\sum_{j=1}^{n}\sigma_{ij}^{2}.\end{equation} Note that we clearly have $\mu^{2}(n)\leq \sigma^{2}(n)$ from these definitions.

For our main results, we will make the following assumptions on the signal in the matrix $P$ and the order of the variance $\sigma^2(n)$ and $\mu^2(n)$. 
\begin{assumption} (``Sufficient signal, sufficient noise")
\label{ass:typical}
Suppose that \begin{align*}c_1 n\geq \lambda_1(P)&\geq \lambda_d(P)\geq c_2 n,\text{ and}\\c_1' n\geq \sigma^{2}(n)&\geq \mu^{2}(n)\geq c_2' n,\end{align*} for some $c_1>c_2>0$ and $c_1'>c_2'>0$.
\end{assumption}

\begin{remark}\label{rem:suff_signal_suff_noise_why} In the case that the matrix $A$ is obtained from real data, which is truncated at a given floating point precision, we will have at least $\sigma^{2}(n)\geq \mu^{2}(n)\geq cn$ for some constant $c>0$ depending on this precision. When the variance is too large, the signal overwhelms the noise, corrupting the computed eigendecomposition. On the other hand, when the variance is very small, early termination may not be optimal. Indeed, in the very low-variance case, since $A$ is likely to be extremely close to $P$, the major source of error is numerical rather than statistical, so incremental reduction of this error can still improve the final quality of the approximation.
\end{remark}

We use the following notation for the spectral decomposition of $P$:
\begin{equation}\label{eq:spec_decomp_P}
P=[U_P|\tilde U_P][S_P\oplus {\bf 0}][U_P|\tilde U_P]^\top
\end{equation}
where $S_P$ is the diagonal matrix of the $d$ non-zero eigenvalues of $P$ and $U_P$ is the matrix of associated eigenvectors. Since $P$ is positive semidefinite, it can be expressed as $XX^\top$, where $X=U_PS_P^{1/2}$. 
Denote the spectral decomposition of $A$ analogously:
\begin{equation}\label{eq:spec_decomp_A}
A=[U_A|\tilde U_A][S_A\oplus \tilde S_A][U_A|\tilde U_A]^\top
\end{equation}
where $S_A$ is the diagonal matrix of the $d$ largest eigenvalues of $A$ (in absolute value) and $U_A$ is the matrix of the associated eigenvectors. We define $\hat{X}=U_A S_A^{1/2}$. In Section~\ref{S:StatisticalErrorBounds}, we state upper and lower bounds on $\|U_A-U_P W\|_F$ and $\|\hat{X}-XW\|_F$, where $W$ is an orthogonal matrix.  



Our approach to determining an appropriate error tolerance for random matrix eigendecompositions is interwoven with the matrix size; in other words, the error tolerance $\epsilon$ is allowed to depend on $n$, and we consider the probabilistic implications of this for large $n$.  As such, we employ a strong version of convergence in probability, and we rely on notions of asymptotic order, both of which we describe below.
\begin{definition}\label{overwhelm_prob} If $D_n$ is a sequence of events indexed by $n$, we say that $D_n$ occurs {\em with overwhelming probability} if for any $c>0$, there exists $n_0(c)$ such that for all $n>n_0(c)$, $\PP(D_n) > 1-\frac{1}{n^{c}}$.
\end{definition}
\begin{definition} If $w(n),\alpha(n)$ are two quantities depending on $n$, we will say that {\em $w$ is of order $\alpha(n)$} and use the notation $w(n) \sim \Theta(\alpha(n))$ to denote that there exist positive constants $c, C$ such that for $n$ sufficiently large,
$c\alpha(n) \leq w(n) \leq C \alpha(n).$
We write $w(n) \sim O(\alpha(n))$ if there exists a constant $C$ such that for $n$ sufficiently large, $w(n) \leq C\alpha(n)$.  We write $w(n) \sim o(\alpha(n))$ if $w(n)/\alpha(n) \rightarrow 0$ as $n \rightarrow \infty$, and $w(n)\sim o(1)$ if $w(n) \rightarrow 0$ as $n \rightarrow \infty$. When $\alpha(n)\sim O(w(n)),$ we write $w(n)\sim\Omega(\alpha(n)).$
\end{definition}

\section{Statistical error bounds for eigendecompositions of LPSM matrices}
\label{S:StatisticalErrorBounds}

We now state Theorem~\ref{thm:uaminusup}, our principal result on the order of the statistical error in approximating the matrix of top $d$ eigenvectors $U_P$ of $P$ with the matrix of top $d$ eigenvectors $U_A$ of $A$. Note the asymmetry in the statements of the lower and upper bounds in Theorem~\ref{thm:uaminusup}. For the upper bound, we show that there is {\em some} orthogonal matrix $W$ such that $\|U_A-U_P W\|_F$ can be made small, where $W$ accounts for the nonuniqueness of the eigenvectors: for example, sign changes for a univariate eigenspace, or linear combinations of eigenvectors in a higher-dimensional eigenspace. For the lower bound, we show that for \emph{any} $W$, $\|U_A-U_P W\|_F$ is at least as large as some quantity depending on the variance in the approximation of $P$ by $A$; this discrepancy between the eigenvectors of $P$ and $A$ is a consequence of inherent random error and cannot be rectified by linear transformations. The proofs for all of the results in this section may be found in Appendix~\ref{s:thmuaminusupproof}.

\begin{theorem}
\label{thm:uaminusup}
Let $A$ be an LPSM matrix, satisfying Assumption~\ref{ass:typical}. With overwhelming probability, there exists an orthogonal matrix $W$ such that $$\|U_A-U_P W\|_F\sim O(1/\sqrt{n}),$$ and for any orthogonal matrix $W$, $$\|U_A-U_P W\|_F\sim\Omega(1/\sqrt{n}).$$ In particular, there exists an orthogonal matrix $W$ such that $\|U_A-U_P W\|_F\sim \Theta(1/\sqrt{n}).$

\end{theorem}

\begin{corollary}
\label{c:uaminusupspectral}
In the setting of the previous theorem, the same bounds hold for the spectral norm $\|U_A-U_P W\|_2$.
\end{corollary}

Recall that the scaled eigenvector matrices are defined by $X=U_P S_P^{1/2},$ and $\hat{X}=U_A S_A^{1/2}.$ We also show bounds on the estimation error between the random $\hat{X}$ and the true $X,$ which we will see in Section~\ref{S:Simulations} is useful in network inference.

\begin{theorem}
\label{thm:xhatminusx}
Suppose $A$ is an LPSM matrix, satisfying Assumption~\ref{ass:typical}. Then there is a sequence $\gamma(n)\sim O(\log(n)/\sqrt{n})$ such that with overwhelming probability, there exists an orthogonal matrix $W$ satisfying $$\|\hat{X}-XW\|_{F}\leq C(P)+\gamma(n),$$ where \begin{equation}\label{eq:CP}
C^{2}(P)=\mathrm{tr}(S_P^{-1/2}U_P^\top \EE(A-P)^{2} U_P S_P^{-1/2})
\end{equation}
and $\EE(A-P)^{2}=\mathrm{diag}\left(\left\{\sum_{j}\sigma_{ij}^{2}\right\}_{i}\right).$ With overwhelming probability, for all orthogonal matrices $W$, we also have the associated lower bound $$\|\hat{X}-XW\|_F\geq C(P)-\gamma(n).$$ Moreover, $C(P)$ is of constant order and is bounded away from zero.
\end{theorem}

\begin{corollary}
\label{c:xhatminusxspectral}
In the setting of the previous theorem, $$\frac{1}{\sqrt{d}}C(P)-\gamma(n)\leq \|\hat{X}-XW\|_2\leq C(P)+\gamma(n).$$
\end{corollary}

\section{Numerical error bounds for eigendecompositions of LPSM matrices}\label{S:NumericalErrorBounds}

Given a numerical eigendecomposition algorithm, let $\widehat{S}_k$ be the $k$th iterate of the diagonal matrix of the largest $d$ approximate eigenvalues of $A$ and let $\hU_k$ be the $k$th iterate of the $n \times d$ matrix of the corresponding orthonormal approximate eigenvectors.  As described in Eq. \eqref{eq:tol}, we consider algorithms that terminate 
when the following stopping criterion is achieved:
\begin{equation}\label{eq: irlba_terminal_condition}
\frac{\|A\hU_k - \hU_k \widehat{S}_k\|_2}{\|A\|_2} \leq \epsilon
\end{equation}
where $\epsilon$ is the user-specified error tolerance.
\begin{remark}
As mentioned in Section~\ref{S:Intro}, the stopping criterion above is computationally attractive and appealing to intuition, but has a significant drawback: if $\hU_{k}$ is the matrix whose columns are the eigenvectors corresponding to $\lambda_{d+1},\ldots,\lambda_{2d}$, then \eqref{eq: irlba_terminal_condition} will hold, but the desired approximation of the top $d$ eigenvectors fails. Thus, in our analysis, we consider the alternative stopping criterion
\begin{equation}\label{eq: irlba_terminal_condition_alt}
\frac{\|A\hU_k-\hU_k\widehat{S}_k\|_2}{\|A\|_2}+\|\widehat{S}_k^{-1}\|_2\leq\epsilon
\end{equation}
This criterion addresses complications that can occur because of problematic initializations: specifically, an initialization $\hat{U}_0$ that belongs to an invariant subspace of $A$ other than the one spanned by the top $d$ eigenvalues, or, because of numerical error, appears to belong to such a subspace. As an example, we generate a matrix $A$ with seven eigenvalues in the interval $(10,11),$ and 93 eigenvalues in the interval $(0,1)$. If we call \texttt{irlba}, an implementation of Implicitly Restarted Lanczos Bidiagonalization \citep{baglama}, with an initial vector $v$ whose first seven entries are zero,
the eigenspace associated to the top seven eigenvalues is entirely missed. This still occurs when the entries in the first components of the initial vector are very small but nonzero ($\approx 10^{-100}$). On the other hand, if we run \texttt{irlba} on $QAQ^\top,$ where $Q$ is a generic unitary matrix, with starting vector $Qv$, 
the correct eigenspace is approximated. Such numerical complications are likely to depend on details of the algorithmic implementation, and while we suspect that, in practice, problematic initializations  should rarely occur for commonly-used algorithms, it is difficult to make general statements about their likelihood.

Importantly, we show in Thm. \ref{thm:numapprox} that for our recommended choice of $\epsilon$, when the criterion in \eqref{eq: irlba_terminal_condition_alt} is satisfied, the term $\|\widehat{S}_{k}^{-1}\|_{2}$ is of strictly lower order than the relative error. Thus, off of the set of problematic initializations, the termination conditions in \eqref{eq: irlba_terminal_condition} and \eqref{eq: irlba_terminal_condition_alt} are functionally equivalent. On the other hand, problematic initializations will go undetected by stopping criterion~\eqref{eq: irlba_terminal_condition}, but on such initializations, criterion~\eqref{eq: irlba_terminal_condition_alt} will prevent the algorithm from terminating, generating a flag that the algorithm terminated for reasons other than the achievement of the stopping criterion. This makes it clear to the user that something has gone awry.

 Equation~\eqref{eq: irlba_terminal_condition_alt} also has a significant theoretical advantage: We show that under this condition and our LPSM assumptions, with high probability (depending only on the random nature of $A$), the nearest eigenvalues of $A$ to those of $\hat{S}_k$ are precisely $\lambda_1(A),\ldots,\lambda_d(A),$ and the eigenvectors $\hat{U}_k$ approximate the top $d$ eigenvectors in $U_A$, as desired. Since we make no assumptions about the algorithm used to compute the approximate eigenvectors and eigenvalues besides the stopping criterion, in order to guarantee this attachment between approximation and truth, we need a stronger condition than Equation~\eqref{eq: irlba_terminal_condition}.
 \end{remark}
 
To choose $\epsilon$ ``optimally"---that is, small enough for accuracy but not so small as to squander computational resources---we need to understand how the stopping criterion impacts the separation between the algorithmically-computed $k$th iterate matrix $\hU_k$ and $U_A$, the true matrix of the $d$ largest eigenvectors of $A$. Similarly, we must understand how the algorithm impacts the separation between the true and approximate eigenvalues. We describe both in the following proposition, the proof of which may be found in Appendix~\ref{s:propnumapproxproof}.

\begin{theorem}
\label{thm:numapprox}

Suppose that $A$ is a symmetric matrix whose top $d$ eigenvalues $\lambda_1(A),\ldots,\lambda_d(A)$ satisfy $\lambda_1(A),\ldots,\lambda_d(A) \sim \Theta(n)$ and whose remaining eigenvalues $\lambda_{d+1}(A), \ldots, \lambda_n(A)$ satisfy $|\lambda_{d+1}(A)|,\ldots,|\lambda_{n}(A)| \sim O(\sqrt{n})$. Let $\hat{U}_k,\hat{S}_k$ be approximate matrices of eigenvectors and eigenvalues for $A$, with the diagonal entries of $S_A,\hat{S}_k$ nonincreasingly ordered. Suppose $\epsilon \sim o(1/\sqrt{n})$. Then for $n$ sufficiently large, $$\frac{\|A\hU_k-\hU_k\widehat{S}_k\|_2}{\|A\|_2}+\|\widehat{S}_k^{-1}\|_2\leq \epsilon$$
guarantees that
$$\frac{\|\hat{S}_k-S_A\|_2}{\|A\|_2}\leq \epsilon,$$ and there is an orthogonal matrix $W$ and constant $C>0$ such that $$\|\hat{U}_k-U_A W\|_F< C\epsilon.$$ Moreover, $\|\hat{S}_k^{-1}\|_2\sim O(1/n)$.
\end{theorem}

\section{Optimal numerical tolerance}\label{S:Opt_Tol}

%
%

We are now ready to give our results combining the bounds on the statistical and numerical error, and suggesting an optimal choice of $\epsilon$ in the stopping criterion so that the numerical error is of just smaller order than the statistical error. Indeed, as we will soon see, with high probability, choosing $\epsilon$ any smaller does not change the order of the error in the approximation $\hat{U}_k\approx U_P W$, so the further effort and computation required to achieve this reduction in the numerical error is essentially wasted. The proof of all of the results in this section may be found in Appendix~\ref{s:thmnumtolproof}.

\begin{theorem}\label{thm:numtol} Let $A$ be an LPSM matrix, and let $\hat{U}_k,\hat{S}_k$ be the approximated eigenvectors and eigenvalues of $A$ satisfying Equation~\ref{eq: irlba_terminal_condition_alt}, where $\epsilon$ denotes the error tolerance. Let $C(P)$ be defined as in Equation~\ref{eq:CP}. Then there is a constant $C>0$ and a sequence $\beta(n)\sim O(\log(n)/n)$ such that with overwhelming probability, for $n$ sufficiently large, there exists an orthogonal matrix W such that \begin{equation}\label{eq:num_tol_critical_bound}\frac{C(P)}{\sqrt{\|A\|_2}}-\beta(n)-C\epsilon\leq \|\hat{U}_k-U_P W\|_F\leq \frac{C(P)}{\sqrt{\lambda_d(A)}}+\beta(n)+C\epsilon,\end{equation} where the lower bound holds for all $W$. If $\epsilon\sim o(\|A\|_2^{-1/2}),$ then the lower bound is of order $\|A\|_2^{-1/2}.$ If the rate at which $\epsilon\rightarrow0$ is increased, then the order of the lower bound in Equation~\ref{eq:num_tol_critical_bound} is not improved.
\end{theorem}

\begin{corollary}
\label{c:numtolspectral}
In the setting of the previous theorem, the following inequalities also hold:
$$\frac{C(P)}{\sqrt{d\|A\|_2}}-\beta(n)-C\epsilon\leq \|\hat{U}_k-U_P W\|_2\leq \frac{C(P)}{\sqrt{\lambda_d(A)}}+\beta(n)+C\epsilon.$$
\end{corollary}

\begin{theorem}\label{thm:numtol_xs} Let $A$ be a LPSM matrix, and let $\hat{U}_k \hat{S}_k^{1/2}$ be the numerical approximation of $\hat{X}=U_A S_A^{1/2}$, where $\hat{U}_k,\hat{S}_k$ satisfy Equation~\ref{eq: irlba_terminal_condition_alt} for a given error tolerance $\epsilon$. Let $C(P)$ be defined as in Equation~\ref{eq:CP}. Then there exists a constant $C>0$ and sequence $\beta(n)\sim O(\log(n)/n)$ such that with overwhelming probability, for $n$ sufficiently large, there exists an orthogonal matrix W such that \begin{equation}\label{eq:num_tol_critical_bound_xs}\frac{C(P)}{2\sqrt{\|A\|_2}}-\beta(n)- C\epsilon\leq \frac{\|\hat{U}_k \hat{S}_k^{1/2}-U_P S_P^{1/2} W\|_F}{\sqrt{\|P\|_2}}\leq \frac{2 C(P)}{\sqrt{\|A\|_2}}+\beta(n)+C \epsilon,\end{equation} where the lower bound holds for all $W$. If $\epsilon\sim o(\|A\|_2^{-1/2}),$ then the lower bound is of order $\|A\|_2^{-1/2}.$ If the rate at which $\epsilon\rightarrow0$ is increased, then the order of the lower bound in Equation~\ref{eq:num_tol_critical_bound} is not improved.
\end{theorem}

\begin{corollary}
\label{c:numtol_xsspectral}
In the setting of the previous theorem, the following inequalities also hold:
$$\frac{C(P)}{2\sqrt{d\|A\|_2}}-\beta(n)-C\epsilon\leq \frac{\|\hat{U}_k\hat{S}_k^{1/2}-U_P S_P^{1/2}W\|_2}{\sqrt{\|P\|_2}}\leq \frac{2C(P)}{\sqrt{\|A\|_2}}+\beta(n)+C\epsilon.$$
\end{corollary}


In particular, Theorems~\ref{thm:numtol} and \ref{thm:numtol_xs} ensure that when $\epsilon<<\|A\|_{2}^{-1/2}$, the statistical error in the approximation of $U_P W$ (respectively, of $U_P S_P^{1/2} W$) by $\hat{U}_k$ (respectively, $\hat{U}_k \hat{S}_k^{1/2}$) dominates the numerical error, so with high probability, the additional computational resources used to achieve this small numerical error have been squandered.

\begin{remark} In practice, for the finite sample case, the constants in Theorem~\ref{thm:numtol} cannot be determined prior to computation. Upper bounds for certain constants can be given, but they are typically far from sharp; see \cite{tang14:_semipar}. Nevertheless, we can suggest a large-sample ``rule-of-thumb." Namely, under our model assumptions,
choosing $\epsilon$ to be of just slightly smaller order than $1/\sqrt{\|A\|_2}$ allows us to account for the unknown constants. Thus, a potential  heuristic is to let  $\epsilon$ be approximately $1/[\log(\log(n)) \sqrt{\|A\|_2}]$; that is, we simply want $\epsilon$ to be of just smaller order than $1/\sqrt{\|A\|_2}$, and the $\log(\log(n))$ factor is but one of many that would allow us to achieve this. From the point of view of implementation, the reader may wonder how to impose a tolerance of order $1/\sqrt{\|A\|_2}$ without actually calculating the spectral norm of $A$ itself. To address this, note that the maximum absolute row sum of $A$, $\delta(A)$, is both inexpensive to compute and serves as an upper bound for the spectral norm, so $1/\sqrt{\delta(A)}$ can be employed as a conservative tolerance for the computation of the largest-magnitude eigenvalue of $A$. Once this top eigenvalue is computed, subsequent computations can proceed with $1/\sqrt{\|A\|_2}$ as the tolerance.   
\end{remark}
\section{Simulations}\label{S:Simulations}

In our simulations, we examine the impact of optimal stopping for inference in statistical networks, particularly spectral decompositions of adjacency matrices for random dot product graphs (RDPGs).  RDPGs are independent-edge random graphs in which each vertex has an associated {\em latent position}, which is a vector in some fixed, finite-dimensional Euclidean space; the probability of a connection between two vertices is given by the inner product of their latent position vectors. Random dot product graphs are an example of the more general latent position random graphs \citep{Hoff2002}, and as we delineate below, the popular stochastic block model \citep{Holland1983} can be interpreted as an RDPG. Formally,
\begin{definition}\label{def:RDPG}[Random Dot Product Graph (RDPG)] 
	Let $\Omega$ be the subset of $\mathbb{R}^{d}$ such that for any two elements $X_1, X_2 \in \Omega$, $X_1^{\top} X_2 \in [0,1]$.
	Let $X=[X_1 \mid \cdots \mid X_n]^{\top}$ be a $n \times d$ matrix whose rows are elements of $\Omega$. 
	Suppose $A$ is
	a random adjacency matrix given by
	\begin{equation*}
	\Pr[A|X]=
	\prod_{i<j}(X_i^{\top}X_j)^{A_{ij}}(1-X_i^{\top}X_j)^{1-A_{ij}}
	\end{equation*}
	We then write $A \sim \mathrm{RDPG}(X)$ and say that $A$ is the adjacency 
	matrix of a {\em random dot product graph} with {\em latent position} $X$
	{\em of rank} at most $d$.
	\end{definition}
	
When $A\sim \mathrm{RDPG}(X),$ it is necessarily an LPSM matrix. As described previously, given $X$, the probability $p_{ij}$ of adjacency between vertex $i$ and $j$ is simply $X_i^{\top}X_j$, the dot product of the associated latent positions $X_i$ and $X_j$. We define the matrix $P=(p_{ij})$ of such probabilities by $
	P=XX^{\top}$.  We will also write $A \sim
\mathrm{Bernoulli}(P)$ to represent that the existence of an
edge between any two vertices $i,j$, where $i>j$, is a Bernoulli
random variable with probability $p_{ij}$; edges are independent. We
emphasize that the graphs we consider are undirected and
have no self-edges. A vital inference task for RDPGs is the estimation of the matrix of latent positions $X$ from a single observation of a suitably-large adjacency matrix $A$. There is a wealth of recent literature on the implications of latent position estimation for vertex classification, multisample network hypothesis testing, and multiscale network inference. See, for example, \cite{STFP-2011, athreya2013limit, lyzinski13:_perfec, hsbm, tang14:_semipar}.



For clarity in the case of simulations, we focus on  the stochastic block model, a graph in which vertices are partitioned into $k$ separate blocks, and the probability of a connection between two vertices is simply a function of their block memberships. Thus a stochastic block depends on a vector $\tau$ of vertex-to-block assignments and a $k \times k$ block probability matrix $B$ whose entries $B_{rs}$ give the probability of connection between any vertex in block $r$ and any vertex in block $s$. A stochastic block model whose block probability matrix is positive semidefinite can be regarded as a random dot product graph whose latent position matrix has $k$ distinct rows, and the spectral decomposition methods we apply to random dot product graphs can be used to infer $B$ and $\tau$ for a stochastic block model as well \citep{STFP-2011}.

In the case of real data,  we address community detection for a subset of the YouTube network, drawn from the Stanford Network Analysis Project (SNAP).\footnote{For the particular YouTube data on which this is based, see the Stanford Network Analysis Project at \texttt{http://snap.stanford.edu/data/index.html}}  
Our code for simulations can be found at
\texttt{www.cis.jhu.edu/}$\sim$\texttt{parky/Tolerance/tolerance.html}.

First, we investigate stopping criteria for eigendecompositions of stochastic block model graphs. In Figure \ref{fig:sbm900_9000_err}, Panel (a), we generate several instantiations of stochastic block model graphs, each with $n=900$ vertices, and equal size blocks and $3\times3$ block probability matrix whose entries are $0.05$ on the diagonal and $0.02$ on the off-diagonal.  

We recall that the {\em Procrustes} error between two $n \times d$ matrices $X$ and $Y$ is given by
$$\min_{O \in \mathcal{O}^{d \times d}}  \|X -Y O\|_F$$
	where $\mathcal{O}^{d \times d}$ denotes the collection of $d \times d$ orthogonal matrices with real entries. That is, the Procrustes distance between the two matrices is zero if there exists an orthogonal transformation of the row vectors that renders one matrix equal to the other. The non-identifiability inherent to a RDPG (namely, the fact that $(XW)(XW)^{\top}=XX^{\top}$) makes such a Procrustes measure the appropriate choice when inferring latent positions associated to random dot product graphs, since they can only be inferred up to an orthogonal transformation.
	
Using the IRLBA algorithm to compute the first three eigenvectors of the associated adjacency matrices, we plot the average Procrustes error (with standard error bars) between the top three estimated eigenvectors and the true eigenvectors of the associated $P$-matrix, i.e. $P=ZBZ^{\top}$, where $Z$ is the $n \times 3$ matrix of block assignments. We observe that when the numerical tolerance is of order $2^{-6}$, this error stabilizes; note that when $n=900$, $1/[\log(\log(n))\sqrt{n}] \approx 2^{-6}$. The default tolerance used in this implementation of IRLBA, however, is $10^{-6}$, which is roughly $2^{-20}$, and stopping at our specified tolerance reduces the number of iterations by a factor of 2 without negatively affecting the results. For much larger stochastic block models, the corresponding computational reduction can be substantial.
\begin{figure}[h]
  \centering
  \subfloat[]{\includegraphics[width=0.4\columnwidth]{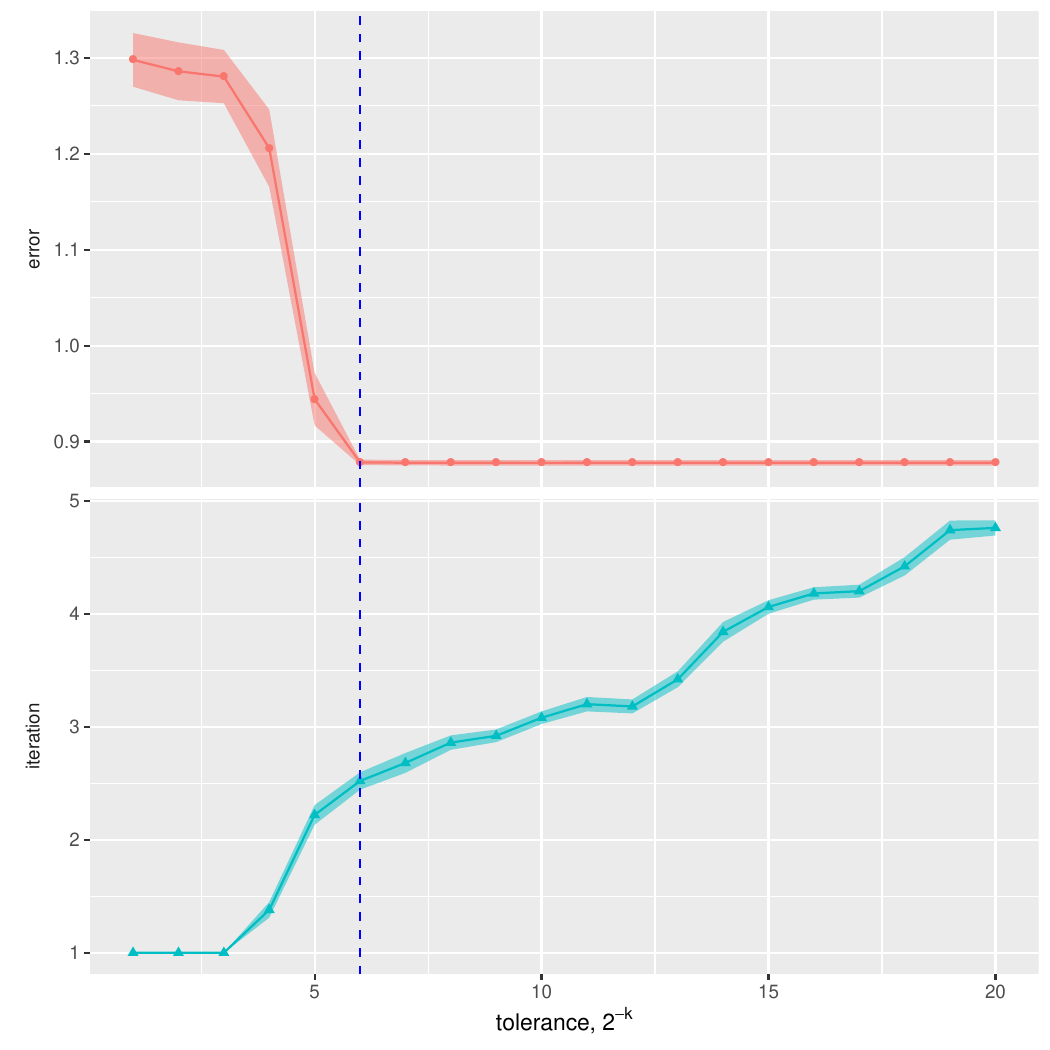}\label{parta}}
~~~
\subfloat[]{\includegraphics[width=0.48\columnwidth]{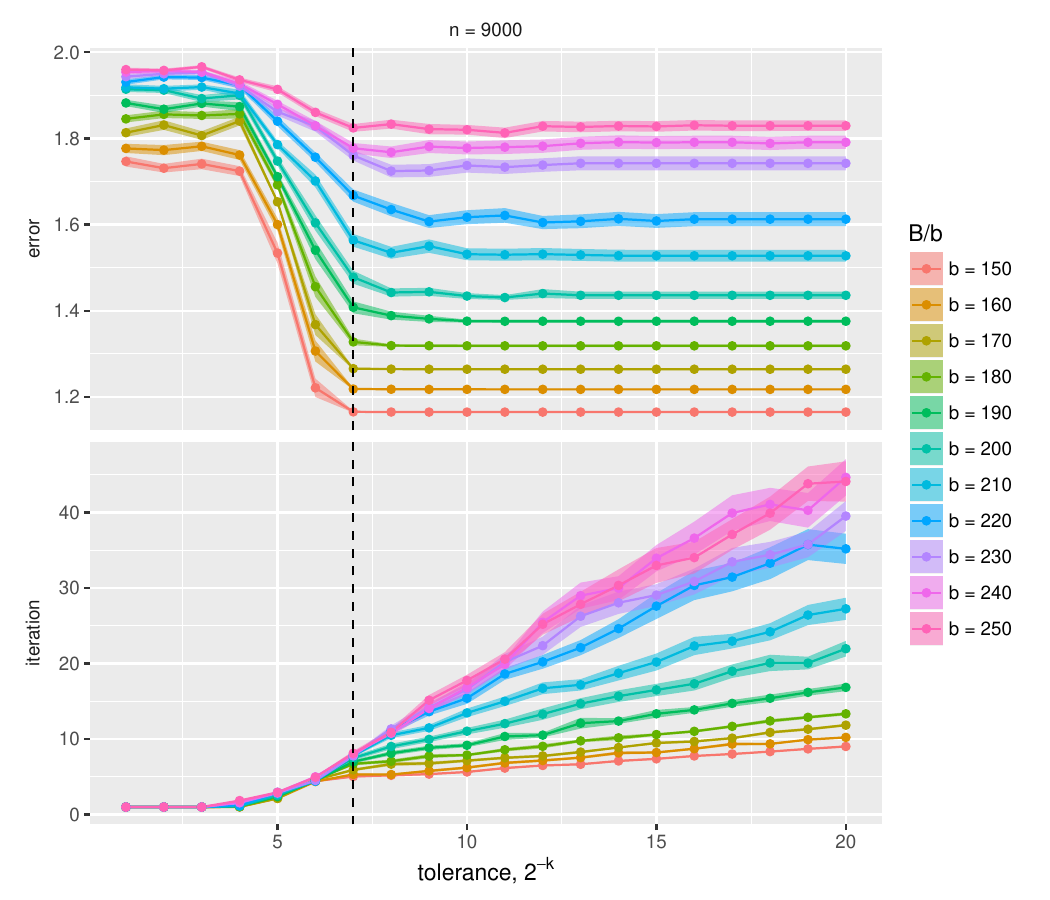}\label{partb}}

  \caption{(Panel a) Top plot represents average Procustes distance between top three approximate eigenvectors of the adjacency matrix and the associated eigenvectors of the probability matrix, for a three-block stochastic block model, plotted against the numerical error tolerance for the eigendecomposition. Bottom plot shows number of iterations. (Panel b) Top plot is average Procrustes distance, with error bars, between approximate eigenvectors of $A$ and the associated eigenvectors of $P$ for SBMs with different $B$ which are scalings of one another by a factor $b$, plotted against the numerical error tolerance for the eigendecomposition. Our heuristic of $1/[\log(\log(n))\sqrt{n}]$ for optimal tolerance performs well across a range of parameters. Bottom plot shows number of iterations.}
  \label{fig:sbm900_9000_err}
\end{figure}

In Figure \ref{fig:sbm900_9000_err} Panel (b), we repeat this procedure for a stochastic block model with $n=9000$ vertices and several different block probability matrices $B$, which are scalings of one another. Observe that our stopping criterion of $1/[\log(\log(n))\sqrt{n}]$ is appropriate across the whole range of these different block probability matrices, saving computation without degrading the results.

Finally, we consider the impact of our stopping criterion on the inference task of community detection in a Youtube network. Youtube networks generally exhibit hierarchical community structure (see \citet{hsbm}) and can be modeled by an SBM/RDPG with a potentially large number of blocks. 
For the task of community detection in a Youtube network, we wish to cluster the rows of the matrix of eigenvectors associated to the adjacency matrix of a Youtube network. We use $K$-means clustering with silhouette width  \citep{rousseeuw_silhouette}, the latter of which is a measure of the rectitude of any particular clustering, to determine the optimal number of clusters. 
The Youtube network we examine has 1,134,890 nodes and 2,987,624 edges, and the spectral decomposition of the adjacency matrix proceeds via IRLBA, with \citet{zhu06:_autom} used to choose the estimated embedding dimension of $\hat{d}=26$. For each $k=1,2, \cdots$ and corresponding tolerance $2^{-k}$, IRLBA generates approximations $\hS_k$ for the the top $\hat{d}=26$ eigenvalues of $A$ and $\hU_k$ for the matrix of associated eigenvectors.  Next, we use $K$-means clustering, with silhouette width to determine the optimal number of clusters, to cluster the rows of the matrix $\hat{U}_k$. 
  We conduct this clustering procedure again, now computing the spectral decomposition with the default tolerance in IRLBA of $10^{-6}$. Ten iterations are run for each of the above clustering methods, and the top plot of Figure \ref{fig:youtube} shows the mean Adjusted Rand Index (ARI), plotted with standard error, for these two clusterings.  We see that at tolerance $2^{-10} \approx 0.00097$, the ARI is nearly 1. Morever, the ARI is very close to 0.95 by tolerance $2^{-7}$. Our heuristic of $1/[\log(\log(n))\sqrt{n}]$ for $n=1134890$ yields a tolerance of $0.000356$.  All of these tolerances---$2^{-7}, 2^{-10}, 3.5\times 10^{-4}$---are of significantly larger magnitude than the default tolerance in IRLBA, which for this implementation is $10^{-6}$ (and in some other implementations can be $10^{-5}$.) This suggests that stopping earlier can save computational time without negatively impacting subsequent inference. In the second plot from the top in Figure \ref{fig:youtube}, we see how the behavior of the ARI for successive pairs $(2^{-k}, 2^{-(k+1)})$ of tolerances.  The bottom two plots in Figure \ref{fig:youtube} indicate the number of iterations and the elapsed time (in seconds) required at each tolerance.
\begin{figure}[!ht]
	\begin{center}
		\includegraphics[width=0.5\columnwidth]{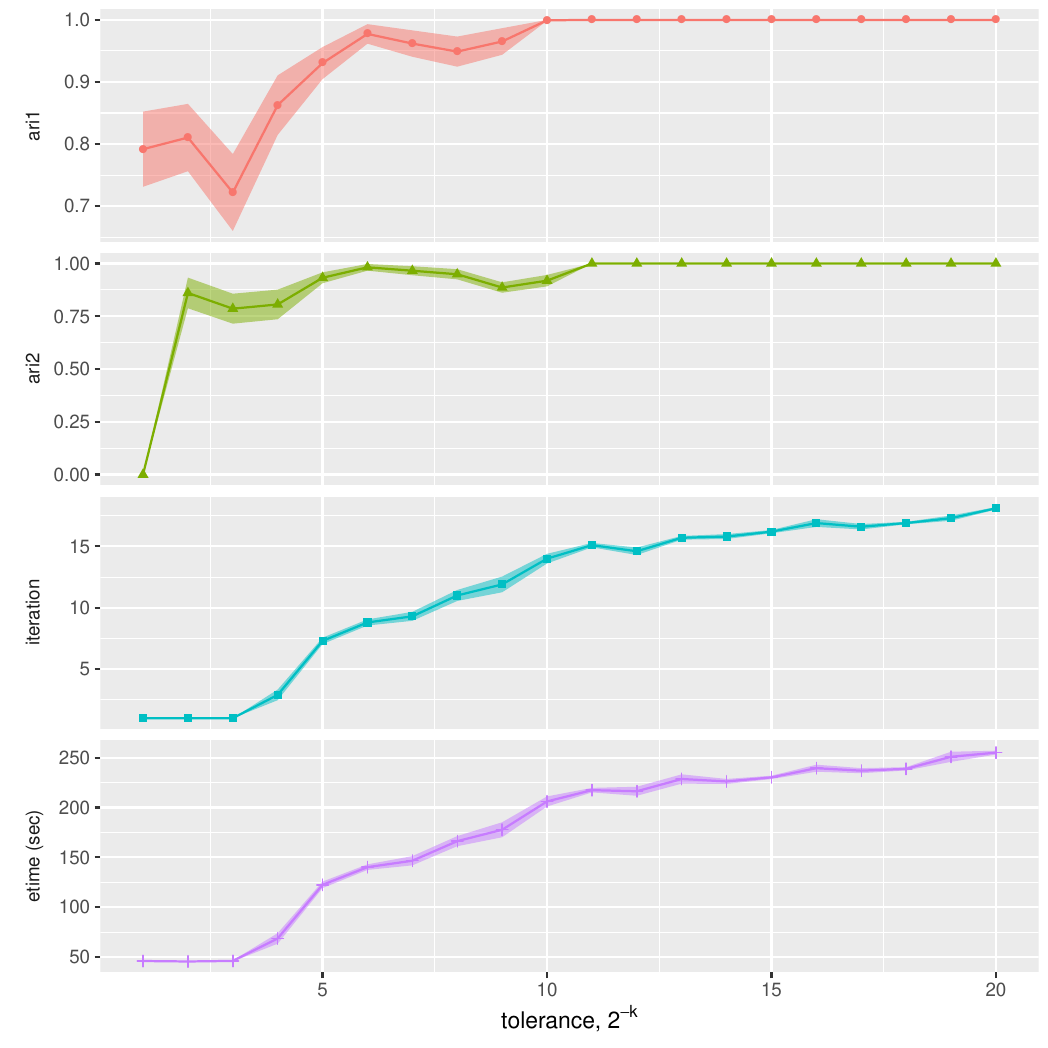}
	\end{center}
	\caption{ARI for comparison of clusterings of the Youtube network with default tolerance in IRLBA and with tolerance $2^{-k}$ (top plot); ARI comparing tolerance of $2^{-k}$ and $2^{-(k+1)}$ (second plot), and the number of iterations and time in seconds (third and fourth plots). 
Our heuristic for the optimal tolerance is $1/[\log(\log(n))\sqrt{n}])=0.00035\approx 2^{-11.5}$ in this case. The ARI is already very close to 0.95 by tolerance $2^{-7}$, which is significantly larger than the default tolerance ($10^{-6}$) in IRLBA.}
	\label{fig:youtube}
\end{figure}

\section{Acknowledgments}
The authors thank John Conroy, Donniell Fishkind, and Joshua Cape for helpful comments, as well as the anonymous referees whose suggestions improved the manuscript.  This work is partially supported by the D3M, XDATA, GRAPHS and SIMPLEX programs of the Defense Advanced Research Projects Agency (DARPA).
\bibliographystyle{plainnat}
\bibliography{biblio}

\appendix

\section{Proofs of Results}
\label{AS:Proofs}


\subsection{Proofs of Theorems~\ref{thm:uaminusup} and \ref{thm:xhatminusx}}
\label{s:thmuaminusupproof}

In this section, we prove our results about the statistical error incurred when using the (exact) eigendecomposition of the observed matrix $A$ to approximate the eigendecomposition of its mean matrix $P$. 

\begin{lemma}
\label{lem:aminusporder}
Suppose $A$ is an LPSM matrix, and suppose Assumption~\ref{ass:typical}. Let $c>0$, and let $\delta>0$ satisfy $n^{-c}<\delta<1/2.$ There is a constant $C$ such that for large enough $n$, $$\PP[\|A-P\|\geq C\sqrt{n\log(n)}]\leq 2\delta.$$ \end{lemma}


\begin{proof}
We argue as in \cite{oliveira2009concentration}. Define for all $1\leq i\leq j\leq n$ the order-$n$ symmetric matrices $$E_{ij}=\begin{cases}e_i e_i^{T}&\text{ if }i=j,\\e_i e_j^\top+e_j e_i^\top&\text{ otherwise},\end{cases}$$ where $e_i$ are the standard unit vectors. Then naturally $A=\sum_{i\leq j}A_{ij}E_{ij},$ and likewise for $P$, so $$A-P=\sum_{i\leq j}(A-P)_{ij}E_{ij}=:\sum_{i\leq j}\Delta_{ij}.$$ The matrices $\Delta_{ij}$ are independent with mean zero and clearly satisfy the bound $\|\Delta_{ij}\|_2\leq \beta\|E_{ij}\|_2=\beta,$ the latter equality following since the eigenvalues of $E_{ij}$ all belong to $\{-1,0,1\}$, and $E_{ij}\neq0$. Since the $\Delta_{ij}$ are symmetric, we see that $$\EE[\Delta_{ij}\Delta_{ij}^\top]=\EE[\Delta_{ij}^{2}]=\sigma_{ij}^{2}E_{ij}^{2}=\begin{cases}\sigma_{ii}^{2}E_{ii}&\text{if }i=j\\ \sigma_{ij}^{2}(E_{ii}+E_{jj})&\text{otherwise.}\end{cases}$$ 
This gives $$\left\|\sum_{i\leq j}\EE[\Delta_{ij}^{2}]\right\|_2=\max_{i}\sum_{j}\sigma_{ij}^{2}=\sigma^{2}(n),$$ the last equality following from the definition in Equation~\ref{eq:sigmandef}. Applying the matrix Bernstein inequality (see Theorem 1.6 \cite{tropp}) shows that for all $t\geq 0$, $$\PP[\|A-P\|_2\geq t]\leq 2n\exp\left(\frac{-t^{2}/2}{\sigma^{2}(n)+\beta t/3}\right).$$ Choosing $t=(1+\sqrt{3})\sqrt{\sigma^{2}(n)\log(n/\delta)}=C\sqrt{n\log(n)}$ in the above completes the proof.

\end{proof}




\begin{lemma}
\label{lem:closestorthogonal}
Let $A$ be an LPSM matrix, and suppose Assumption~\ref{ass:typical} holds. For $c>0$, $n^{-c}<\delta<1/2,$ there is a positive constant $C$ such that for large enough $n$, $$\PP[\|U_A U_A^\top-U_P U_P^\top\|_2\geq C\sqrt{\log(n)/n}]\leq 2\delta.$$
Let $W_1\Sigma W_2^\top$ be the singular value decomposition of $U_P^\top U_A$, and let $W^*=W_1 W_2^\top$. Then there is a constant $C$ such that for large enough $n$, $$\PP[\|U_P^\top U_A-W^*\|_F\geq C\log(n)/n]\leq 2\delta.$$ 

There is a constant $C$ such that for large enough $n$, $$\PP[\|U_P^\top (A-P)U_P\|_{F}\geq C\sqrt{\log(n)}]\leq \delta.$$

There are constants $C_1,C_2,C_3>0$ such that for large enough $n$, \begin{align*}\PP[\|W^* S_A-S_P W^*\|_F\geq C_1\log(n)]&\leq 3\delta,\\
\PP[\|W^* S_A^{1/2}-S_P^{1/2}W^*\|_F\geq C_2\log(n) n^{-1/2}]&\leq 3\delta,\\
\PP[\|S_P^{-1/2}W^*-W^* S_A^{-1/2}\|_F\geq C_3\log(n) n^{-3/2}]&\leq 3\delta.\end{align*}


\end{lemma}

The proofs of these lemmas follow immediately from the proofs of Prop. 16 and Lem. 17 in \citep{hsbm}, so we omit them here.

\begin{lemma}
\label{lem:leadingorder}
Let $A$ be an LPSM matrix, and suppose Assumption~\ref{ass:typical}. Let $c>0, n^{-c}<\delta<1/2$. Then there is a sequence $\gamma(n)\sim O(\log(n)/\sqrt{n})$ such that for large enough $n$, with probability at least $1-3\delta$, there exists a rotation matrix $W$ such that $$\|\hat{X}-XW\|_{F}=\|(A-P)U_P S_P^{-1/2}\|_F+\gamma(n).$$
\end{lemma}

This lemma follows as in Theorem~18 of \cite{hsbm}.

\begin{lemma}
\label{lem:dominanttermconc}
 Let $A$ be an LPSM matrix, and let $V=U_P S_P^{-1/2}.$ 
If Assumption~\ref{ass:typical} holds, and for $c>0$, $n^{-c}<\delta<1/2,$ then there is a constant $C>0$ such that for large enough $n$, $$\PP[|\|(A-P)V\|_F-C(P)|\geq C \log(n)/\sqrt{n}]\leq 4\delta,$$ where $C(P)^2=\EE[\|(A-P)V\|_F^2]\sim \Theta(1).$

\end{lemma}

\begin{proof}
Let $A'$ be an arbitrary symmetric matrix with all entries in the intervals $[\alpha_{ij},\alpha_{ij}+\beta],$ and for each $1\leq r\leq s\leq n$, define $Z_{rs}=\|(A^{(r,s)}-P)V\|_{F}^{2},$ where $A^{(r,s)}$ agrees with $A$ in every entry except the $(r,s)$ and $(s,r)$th ones, where it equals $A'_{rs}$.
Expanding the formula for $Z=\|(A-P)V\|_F^2$, and using symmetry of $A-P$, we see that $$Z=\|(A-P)V\|_{F}^{2}=\sum_{i,j,k,\ell}(A-P)_{ik}(A-P)_{i\ell}V_{kj}V_{\ell j}=\sum_{i,k,\ell}(A-P)_{ik}(A-P)_{i\ell}(VV^{T})_{\ell k}.$$ Since all terms which do not feature $A_{rs}$ remain the same when $A^{(r,s)}$ is introduced, we see that for $r\neq s$, 
\begin{align*}
&Z-Z_{rs}\\
=&2\sum_{\ell}\left[(A-P)_{rs}(A-P)_{r\ell}(VV^{T})_{\ell s}-(A'-P)_{rs}(A^{(r,s)}-P)_{r\ell}(VV^{T})_{\ell s}\right]\\&+2\sum_{\ell}\left[(A-P)_{sr}(A-P)_{s\ell}(VV^{T})_{\ell r}-(A'-P)_{sr}(A^{(r,s)}-P)_{s\ell}(VV^{T})_{\ell r}\right]\\
=&2(A-A')_{rs}\left[\phantom{\int}\!\!\!\![(A-P)VV^{T}]_{rs}+[(A-P)VV^{T}]_{sr}+(A'-P)_{rs}((VV^{T})_{ss}+(VV^{T})_{rr})\right].
\end{align*}
Now since $|(A-A')_{rs}|$ and $|(A'-P)_{rs}|\leq \beta,$ we have for $r\neq s$ $$(Z-Z_{rs})^{2}\leq 16\beta^{2}\left[\phantom{\int}\!\!\!\![(A-P)VV^{T}]_{rs}^{2}+[(A-P)VV^{T}]_{sr}^{2}+\beta^{2}[(VV^{T})_{ss}^{2}+(VV^{T})_{rr}^{2}]\right],$$ and for $r=s$ we obtain $$(Z-Z_{rr})^{2}\leq 8\beta^{2}([(A-P)VV^{T}]_{rr}^2+\beta^{2}(VV^{T})_{rr}^{2}).$$ Summing over $r$ and $s$, we see that \begin{align*}
\sum_{r\leq s}(Z-Z_{rs})^{2}\leq& 16\beta^{2}\sum_{r,s}[(A-P)VV^{T}]_{rs}^{2}+8\beta^{4}\sum_{r}(VV^{T})_{rr}^{2}+16\beta^{4}\sum_{r< s}((VV^{T})_{ss}^{2}+(VV^{T})_{rr}^{2})\\\leq& 16\beta^{2}\|(A-P)VV^{T}\|_{F}^{2}+8\beta^{4}\sum_{r,s}((VV^{T})_{rr}^{2}+(VV^{T})_{ss}^{2})\\\leq& 16\beta^{2}Z\|V\|_2^2+16\beta^{4}n\|\mathrm{diag}(VV^{T})\|_{F}^{2}\leq \frac{16\beta^2}{\lambda_d(P)}Z+\frac{16\beta^4 nd}{\lambda_d(P)^2},\end{align*} since $VV^{T}=U_P S_P^{-1} U_P^{T}.$

Applying Theorems 5 and 6 in \cite{boucheron_massart_lugosi_conc_aop_2003}, we get the following bound for $t>0$, $a=16\beta^2/\lambda_d(P)$, and $b=16\beta^4 n d/\lambda_d(P)^2:$ \begin{equation}\label{eq:domtermconcbound}\PP[|Z-\EE[Z]|>t]\leq 2\exp\left(\frac{-t^2}{4a\EE[Z]+4b+2at}\right).\end{equation} Note that there exist constants $c_a,c_b>0$ such that $a\leq c_a/n$ and $b\leq c_b/n$.


We expand $\EE\|(A-P)V\|_F^2=\EE(\mathrm{tr}(V^{T}(A-P)^{2}V))=\mathrm{tr}(V^{T}\EE(A-P)^{2}V),$ and observe that \begin{align*}
\EE(A-P)^{2}_{ij}&=\EE\left[\sum_{k}(A-P)_{ik}(A-P)_{kj}\right]\\
&=\begin{cases}\sum_{k}\sigma_{ik}^{2}&\text{if }i=j\\0&\text{otherwise,}\end{cases}
\end{align*}
since $i\neq j$ implies that $(A-P)_{ik},(A-P)_{kj}$ are independent, and both have mean zero. In the case $i=j,$ this is simply the definition of the entrywise variances. Then from Assumption~\ref{ass:typical}, we see that in the positive semidefinite ordering, $c_1' n I \succcurlyeq \EE(A-P)^{2}\succcurlyeq c_2' n I,$ which means that 
$$c_2' n \|V\|_{F}^{2}\leq C^{2}(P)\leq c_1' n\|V\|_F^2.$$ But since $V=U_P S_P^{-1/2},$ and $$\|V\|_{F}^{2}=\|S_{P}^{-1/2}\|_{F}^{2}=\sum_{i=1}^{d}\lambda_{i}(P)^{-1}=\Theta(n^{-1})$$ by virtue of Assumption~\ref{ass:typical}, which gives $C^{2}(P)=\Theta(1)$.

Finally, returning to Equation~\eqref{eq:domtermconcbound} and choosing $t=C\log(n)/\sqrt{n}$, we obtain $\PP[|Z-C^2(P)|>t]\leq 2\delta$. Establishing the required bound is now a matter of straightforward algebra.
\end{proof}

We now prove Theorem~\ref{thm:xhatminusx}.
\begin{proof}
From Lemma~\ref{lem:leadingorder}, with probability at least $1-3\delta$, $$\|\hat{X}-XW\|_{F}=\|(A-P)V\|_F+O(\log(n)/\sqrt{n}),$$ and by Lemma~\ref{lem:dominanttermconc}, with probability at least $1-4\delta$, $$\|(A-P)V\|_{F}=C(P)+O(\log(n)/\sqrt{n}).$$ Therefore, with probability at least $1-7\delta,$ $$\|\hat{X}-XW\|_F=C(P)+O(\log(n)/\sqrt{n}).$$
\end{proof}

We now prove Theorem~\ref{thm:uaminusup}.

\begin{proof}
We begin by computing \begin{align*} U_A&= U_A S_A^{1/2} S_A^{-1/2}=\hat{X}S_A^{-1/2}=U_P S_P^{1/2}W^* S_A^{-1/2}+(\hat{X}-XW)S_A^{-1/2}\\
&=U_P W^*+U_P S_P^{1/2}(W^* S_A^{-1/2}-S_P^{-1/2}W^*)+(\hat{X}-XW)S_A^{-1/2}.
\end{align*}
Then \begin{align*}|\|U_A-U_P W^*\|_F&-\|(\hat{X}-XW)S_A^{-1/2}\|_F|\\\leq&	\|S_P^{1/2}(W^* S_A^{-1/2}-S_P^{-1/2}W^*)\|_F\\\leq& \sqrt{\|P\|_2}\|S_P^{-1/2}W^*-W^* S_A^{-1/2}\|_F.\end{align*} Expanding $\|(\hat{X}-XW)S_A^{-1/2}\|_F^{2},$ we see that $$\|(\hat{X}-XW)S_A^{-1/2}\|_F^{2}=\sum_{i,j}\frac{(\hat{X}-XW)_{i,j}^{2}}{\lambda_j(A)},$$ so \begin{equation}\label{eq:xhatminusxwsanegonehalf}\frac{\|\hat{X}-XW\|_F^2}{\lambda_1(A)}\leq\|(\hat{X}-XW)S_A^{-1/2}\|_F^{2}\leq \frac{\|\hat{X}-XW\|_F^2}{\lambda_d(A)}.\end{equation} Taking square roots and applying the bounds $\sqrt{\lambda_1(A)}\leq \sqrt{\|P\|_2}+\|A-P\|_2/(2\sqrt{\|P\|_2}),$ $\sqrt{\lambda_d(A)}\geq \sqrt{\lambda_d(P)}-\|A-P\|_2/\sqrt{\lambda_d(P)}$ yields the stated inequalities.

Under Assumption~\ref{ass:typical}, we see that $\|P\|_2,\lambda_d(P)=\Theta(n);$ from Lemma~\ref{lem:aminusporder}, $\|A-P\|_2=O(\sqrt{n\log(n)});$ from Lemma~\ref{lem:closestorthogonal}, $\|S_P^{-1/2}W^*-W^* S_A^{-1/2}\|_F=O(\log(n) n^{-3/2});$ and from Theorem~\ref{thm:xhatminusx}, $\|\hat{X}-XW\|_F=\Theta(1).$ Then the lower and upper bounds are respectively $\Theta(1/\sqrt{n})-O(\log(n)/n)$ and $\Theta(1/\sqrt{n})+O(\log(n)/n),$ both of which are $\Theta(1/\sqrt{n})$, which completes the proof.
\end{proof}

\subsection{Proof of Thm.~\ref{thm:numapprox}}
\label{s:propnumapproxproof}



We make use of the following known results on matrix decompositions.

\begin{theorem}(\citet[Theorem~4.2.15]{stewart}) \label{Stewart} Let the Hermitian matrix $A$ have the spectral representation $A=JLJ^{\top} + YMY^{\top}$, 
	where the matrix $[J \, Y]$ is unitary.  Let the orthonormal matrix $Z$ be of the same dimensions as $J$, 
and let $N$ be a Hermitian matrix, and let $\lambda(M)$ and $\lambda(N)$ denote the spectra of $M$ and $N$. Let $\rho$ denote the minimum distance between any element in $\lambda(M)$ and any element in $\lambda(N)$, and suppose $\rho>0$.  Then
	\begin{equation}
	\|\sin \Psi(\mathcal{R}(J), \mathcal{R}(Z))\|_F \leq \frac{\|AZ-ZN\|_F}{\rho}
	\end{equation}
	where $\mathcal{R}(J)$ and $\mathcal{R}(Z)$ denote the eigenspaces of the  matrices $J$ and $Z$, respectively,
and $\Psi(\mathcal{R}(J), \mathcal{R}(Z))$ denotes the diagonal matrix of canonical angles between them.
\end{theorem}

\begin{theorem}(\citet{kahan}) \label{Kahan} Let $A \in \mathbb{C}^{n \times n}$ and $B \in \mathbb{C}^{d \times d}$ each be Hermitian.  Let $H \in \mathbb{C}^{n \times d}$ have orthonormal columns. 
	Then to the eigenvalues $\mu_1, \mu_2, \dots, \mu_l$ of $B$, there correspond $l$ eigenvalues $\lambda_1, \lambda_2, \cdots, \lambda_l$ of $A$ such that
	$$|\lambda_i -\mu_i| \leq \|AZ-ZN\|_2$$
\end{theorem}

We now prove Thm.~\ref{thm:numapprox}.
\begin{proof}
Suppose $c_1 n\geq \lambda_1(A),\ldots,\lambda_d(A)\geq c_2 n$ for some $c_1>c_2>0,$ and $c_3\sqrt{n}\geq |\lambda_{d+1}(A)|,\ldots,|\lambda_{n}(A)|$. Let $\hat{U}_k,\hat{S}_k$ be approximate matrices of eigenvectors and eigenvalues for $A$, with the diagonal entries of $S_A$ and $\hat{S}_k$ nonincreasingly ordered. Suppose $\epsilon \leq 1/c_4\sqrt{n}$  where $c_4>c_3+(c_1/c_3)$. We first show that if Equation~\ref{eq: irlba_terminal_condition_alt}
holds with such $\epsilon$, then $\hat{S}_k$ correctly estimates the top $d$ eigenvalues of $A$.  To see why, note that Theorem~\ref{Kahan} guarantees that there exist eigenvalues $\lambda_{i_1}(A),\ldots,\lambda_{i_d}(A)$ of $A$ such that 
\begin{equation}\label{eq:Kahan_applied}
|\lambda_j(\hat{S}_k)-\lambda_{i_j}(A)|\leq \|A\hat{U}_k-\hat{U}_k\hat{S}_k\|_2\leq \|A\|_2\epsilon.
\end{equation}
where the latter inequality follows by Eq.\eqref{eq: irlba_terminal_condition_alt}.
But if $i_j>d$ for some $j$, then $|\lambda_j(\hat{S}_k)|\leq |\lambda_{i_j}(A)|+\epsilon\|A\|_2\leq c_3\sqrt{n}+(c_1/c_3)\sqrt{n}<c_4\sqrt{n}$ by the bound on $c_4$, so $|\lambda_j(\hat{S}_k)|^{-1} > 1/(c_4\sqrt{n})$, and $\|\hat{S}_k^{-1}\|_2>\epsilon$, which contradicts Eq.\eqref{eq: irlba_terminal_condition_alt}.
Now, by Eq.\eqref{eq:Kahan_applied}, since $1\leq i_j\leq d$, we observe that
$$\|\hat{S}_k-S_A\|_2 =\max_j |\lambda_j(\hat{S}_k)-\lambda_j(A)|\leq \max_j |\lambda_j(\hat{S}_k)-\lambda_{i_j}(A)|\leq \|A\hat{U}_k-\hat{U}_k\hat{S}_k\|_2\leq \|A\|_2\epsilon,$$ which proves the first statement. We now show that there is a constant $C>0$ and orthogonal matrix $W$ such that
$$\|\hat{U}_k-U_A W\|_F<C\epsilon.$$


Since we have shown that for $j\leq d,$ $\lambda_j(\hat{S}_k)\geq \lambda_j(A)-\epsilon\|A\|_2\geq c_2 n-(c_1/c_3)\sqrt{n}$, we see that for any $r>d,$ $|\lambda_j(\hat{S}_k)-\lambda_r(A)|\geq c_2 n-(c_1/c_3)\sqrt{n}-c_3\sqrt{n}> c_2 n-c_4\sqrt{n}$, which means that $\rho=\min_{j\leq d<r}|\lambda_j(\hat{S}_k)-\lambda_r(A)|\geq c_2 n-c_4\sqrt{n}$. 

By Theorem~\ref{Stewart}, $$\|\sin\Psi(\mathcal{R}(U_A),\mathcal{R}(\hat{S}_k))\|_F\leq\frac{\|A\hat{U}_k-\hat{U}_k\hat{S}_k\|_F}{\rho}\leq \frac{\sqrt{d}\|A\hat{U}_k-\hat{U}_k\hat{S}_k\|_2}{\|A\|_2}\frac{\|A\|_2}{\rho}\leq \sqrt{d}\epsilon \frac{c_1 n}{c_2 n-c_4\sqrt{n}},$$ which means this bound also holds for $\|U_A U_A^\top-\hat{U}_k \hat{U}_k^\top\|_F.$ 

By  \citet[Proposition~2.1]{rohe2011spectral}, there is an orthonormal matrix $W$ such that
\begin{equation}
\|U_A W-\hU_k\|_F \leq \sqrt{2d}\epsilon \frac{c_1 n}{c_2 n-c_4\sqrt{n}},
\end{equation} and since $c_2 n-c_4\sqrt{n}> (c_2/2) n$ once $n$ is large enough, this upper bound is at most $2\sqrt{2d}(c_1/c_2)\epsilon=:C\epsilon$.
\end{proof}

\subsection{Proofs of Theorems~\ref{thm:numtol} and \ref{thm:numtol_xs}}
\label{s:thmnumtolproof}


We decompose the total error as a sum of statistical and numerical error; our choice of $\epsilon$ guarantees that the statistical error is the dominant term.

\begin{proof}
Throughout this proof, whenever we bound terms involving $A$, we work on the set of high probability where the various supporting bounds hold. From the proof of Theorem~\ref{thm:uaminusup}, $\|U_A-U_P W^*\|_F\in\|(\hat{X}-XW)S_A^{-1/2}\|_F\pm\beta(n)$, and using Equation~\ref{eq:xhatminusxwsanegonehalf}, $$\frac{\|\hat{X}-XW\|_F^2}{\lambda_1(A)}-\beta(n)\leq \|U_A-U_P W^*\|_F^2\leq \frac{\|\hat{X}-XW\|_F^2}{\lambda_d(A)}+\beta(n),$$ so applying Theorem~\ref{thm:xhatminusx} gives us
$$\frac{C(P)}{\sqrt{\|A\|_2}}-\beta(n)\leq\|U_A-U_P W^*\|_F\leq \frac{C(P)}{\sqrt{\lambda_d(A)}}+\beta(n),$$ where $\beta(n)\sim O(\log(n)/n),$ and the lower bound still holds if we replace $W^*$ for any other $W$. Then for some choice of $W$ and $W_1$, we see that $$\hat{U}_k-U_P W W_1=\hat{U}_k-U_A W_1+U_A W_1-U_P W W_1.$$ So if we let $W_1$ be the closest orthogonal matrix to $U_A^\top \hat{U}_k,$ and $W^*$ be the closest orthogonal matrix to $U_P^\top U_A$, then
\begin{align*}
\min_W\|\hat{U}_k-U_P W\|_F 
&\leq \|\hat{U}_k-U_A W_1\|_F+\|U_A W_1-U_P W^* W_1\|_F
\leq C\epsilon+\frac{C(P)}{\sqrt{\lambda_d(A)}}+\beta(n).
\end{align*} 
Arguing similarly, we have for any $W$,
$$ \|\hat{U}_k-U_P W W_1\|_F\geq \|U_A W_1-U_P W W_1\|_F-\|\hat{U}_k-U_A W_1\|_F\geq \frac{C(P)}{\sqrt{\|A\|_2}}-\beta(n)-C\epsilon.$$ Note that the set $\{W W_1: W\text{ is orthogonal}\}=\{W: W\text{ is orthogonal}\}$, so this lower bound holds for all $W$. As soon as $\epsilon\sim o(1/\sqrt{n}),$ since $\beta(n)\sim O(\log(n)/n)$ is $o(1/\sqrt{n}),$ the term $C(P)/\sqrt{\|A\|_2}$ becomes the dominant term, so further reduction in $\epsilon$ makes no difference to the order of the lower bound.

Now we address the scaled case. Let $W_1$ be the closest orthogonal matrix to $U_A^\top \hat{U}_k$ as before, and let $W^*$ be the closest orthogonal matrix to $U_P^\top U_A$ as usual. We see that $$\hat{U}_k\hat{S}_k^{1/2}-U_P S_P W^* W_1=\hat{U}_k \hat{S}_k^{1/2}-U_A S_A^{1/2} W_1+U_A S_A^{1/2} W_1-U_P S_P^{1/2}W^* W_1,$$ where the last two terms are just $(\hat{X}-XW)W_1$ for some orthogonal matrix $W$. From Thm.~\ref{thm:numapprox}, we know that $\|U_A^\top \hat{U}_k-W_1\|_F\leq \|\hat{U}_k-U_A W_1\|_F\leq C\epsilon.$ Then since 
\begin{align*}
S_A W_1-W_1 \hat{S}_k&= S_A U_A^\top \hat{U}_k+S_A(W_1-U_A^\top \hat{U}_k)-W_1 \hat{S}_k
=U_A^\top A\hat{U}_k+S_A(W_1-U_A^\top \hat{U}_k)-W_1 \hat{S}_k\\
&=(U_A^\top \hat{U}_k-W_1)\hat{S}_k+U_A^\top(A\hat{U}_k-\hat{U}_k\hat{S}_k)+S_A(W_1-U_A^\top \hat{U}_k),
\end{align*} we see that
$$\|S_A W_1-W_1\hat{S}_k\|_F\leq C\epsilon( \|\hat{S}_k\|_2+\|S_A\|_2)+\|A\|_2\epsilon,$$ which is further bounded by $\|A\|_2\epsilon(2C+C\epsilon+1).$ Arguing as in Lemma~\ref{lem:closestorthogonal}, we have $$\|S_A^{1/2} W_1-W_1\hat{S}_k^{1/2}\|_F\leq \frac{\|A\|_2\epsilon(2C+C\epsilon+1)\sqrt{\lambda_d(A)}}{2\lambda_d(A)-\|A\|_2\epsilon}\leq C'\sqrt{\|A\|_2}\epsilon.$$

We compute
\begin{align*}
&\hat{U}_k\hat{S}_k^{1/2}-U_A S_A^{1/2}W_1
=\hat{U}_k\hat{S}_k\hat{S}_k^{-1/2}-U_A S_A^{1/2}W_1\\
=&A\hat{U}_k \hat{S}_k^{-1/2}+(\hat{U}_k \hat{S}_k-A\hat{U}_k)\hat{S}_k^{-1/2}-U_A S_A^{1/2}W_1\\
=&AU_A W_1\hat{S}_k^{-1/2}+A(\hat{U}_k-U_A W_1)\hat{S}_k^{-1/2}+(\hat{U}_k \hat{S}_k-A\hat{U}_k)\hat{S}_k^{-1/2}-U_A S_A^{1/2}W_1\\
=&U_A S_A^{1/2}(S_A^{1/2}W_1-W_1 \hat{S}_k^{1/2})\hat{S}_k^{-1/2}+A(\hat{U}_k-U_A W_1)\hat{S}_k^{-1/2}+(\hat{U}_k \hat{S}_k-A\hat{U}_k)\hat{S}_k^{-1/2},
\end{align*} which leads to \begin{align*}
\|\hat{U}_k\hat{S}_k^{1/2}-U_A S_A^{1/2}W_1\|_F&\leq \|A\|_2^{1/2}\|\hat{S}_k^{-1/2}\|_2\|S_A^{1/2}W_1-W_1 \hat{S}_k^{1/2}\|_F\\
&+\|A\|_2\|\hat{S}_k^{-1/2}\|_2\|\hat{U}_k-U_A W_1\|_F+\|\hat{S}_k^{-1/2}\|_2\|A\hat{U}_k-\hat{U}_k\hat{S}_k\|_F.
\end{align*}
From  Thm.~\ref{thm:numapprox}, $\|\hat{S}_k^{-1/2}\|_2\sim O(1/\sqrt{\|A\|_2}),$ so from the bounds proved above and the stopping criterion, we have for some constant $C$ $$\|\hat{U}_k\hat{S}_k^{1/2}-U_A S_A^{1/2}W_1\|_F \leq C\sqrt{\|A\|_2}\epsilon.$$

Returning to the bounds on the total error, we have $$\left|\frac{\|\hat{U}_k\hat{S}_k^{1/2}-U_P S_P W^* W_1\|_F}{\sqrt{\|P\|_2}}- \frac{\|\hat{X}-XW\|_F}{\sqrt{\|P\|_2}}\right|\leq \frac{C\sqrt{\|A\|_2}\epsilon}{\sqrt{\|P\|_2}}.$$ Since from Theorem~\ref{thm:xhatminusx}, $\|\hat{X}-XW\|_F=C(P)+\gamma(n),$ where $\gamma(n)\sim O(\log(n)/\sqrt{n}),$ we see that $$\left|\frac{\|\hat{U}_k\hat{S}_k^{1/2}-U_P S_P W^* W_1\|_F}{\sqrt{\|P\|_2}}-\frac{C(P)}{\sqrt{\|P\|_2}}\right|\leq \beta(n)+C'\epsilon.$$ Now since $|\|A\|_2-\|P\|_2|\leq \|A-P\|_2,$ once $n$ is large enough we have \begin{align*}\sqrt{\|A\|_2}&\leq \sqrt{\|P\|_2}+\frac{\|A-P\|_2}{2\sqrt{\|P\|_2}}\leq 2\sqrt{\|P\|_2},\text{ and}\\
\sqrt{\|A\|_2}&\geq \sqrt{\|P\|_2}-\frac{\|A-P\|_2}{\sqrt{\|P\|_2}}\geq \frac{1}{2}\sqrt{\|P\|_2},
\end{align*}
which yields the final inequalities. Arguing as we did for $\|\hat{U}_k-U_PW\|_F$, and from Theorem~\ref{thm:xhatminusx}, the lower bound holds for any $W$.

\end{proof}

\end{document}